\documentclass[submission,copyright,creativecommons]{eptcs}
\providecommand{\event}{AFL 2023} 

\usepackage[utf8]{inputenc}
\usepackage{amsmath}
\usepackage{amsthm}
\usepackage{amssymb}
\usepackage{tikz}

\usepackage{caption}
\usepackage{subcaption}

\usepackage{comment}

\newtheorem{theorem}{Theorem}
\newtheorem{lemma}[theorem]{Lemma}

\usepackage{iftex}

\ifpdf
  \usepackage{underscore}         
  \usepackage[T1]{fontenc}        
\else
  \usepackage{breakurl}           
\fi

\title{Freezing 1-Tag Systems with States}
\author{Szil\'ard Zsolt Fazekas\thanks{Szil\'ard Zsolt Fazekas was supported by JSPS KAKENHI Grant Number JP23K10976.}
\institute{Akita University\\
Akita, Japan}
\email{szilard.fazekas@ie.akita-u.ac.jp}
\and
Shinnosuke Seki\thanks{Shinnosuke Seki was supported by JSPS KAKENHI Grant-in-Aids for Scientific Research (B) 20H04141 and (C) 20K11672.}
\institute{University of Electro-Communications\\
Tokyo, Japan}
\email{\quad s.seki@uec.ac.jp}
}

\def\titlerunning{Freezing 1-Tag Systems with States}
\def\authorrunning{S.Z. Fazekas, S.Seki}
\begin{document}
\maketitle

\begin{abstract}
We study 1-tag systems with states obeying the freezing property that only allows constant bounded number of rewrites of symbols. We look at examples of languages accepted by such systems, the accepting power of the model, as well as certain closure properties and decision problems. Finally we discuss a restriction of the system where the working alphabet must match the input alphabet.
\end{abstract}




\section{Introduction}

Tag systems are a class of deterministic string rewriting systems. In each step they read out the first letter, say $a$, of the current word along with the succeeding $n-1$ letters (where $n$ is a system parameter), refer to the system's transition table (function) $\delta$ with the letter read as a key, and append the word $\delta(a)$ at the end of the current word. The system stops if the key is a special halting symbol or if there are less than $n$ letters left in the word. A system in this class is often called an $n$-tag system including the value of $n$ explicitly. 
It is well known that $n$-tag systems are capable of simulating Turing machines for any $n \ge 2$  (see, e.g., \cite{CockeM1964, Cook2004}). 
The initial definition is stateless, has no additional storage or intricate rules describing its dynamics; hence such systems are ideal for being simulated in other computational models such as 1D cellular automata or molecular computation models, particularly in their early stage of development wherein knowledge and techniques for programming are yet to be developed. 

Recently, the oritatami model of RNA co-transcriptional folding has been introduced~\cite{GearyMSS2018}. In this model, an abstract RNA sequence folds upon itself greedily while being synthesized in order to achieve various computational tasks \textit{in-silico}. These tasks are usually relatively simple from the viewpoint of computational complexity theory, such as, the detection of some specific molecule for gene expression regulation~\cite{WattersSYLL2016}.

The cyclic variant of tag systems (cts) introduced by Cook \cite{Cook2004} has been simulated in the oritatami model in order to prove its Turing completeness \cite{GearyMSS2018}; periodicity supposed for the sequence to be folded in the oritatami computation also favored this variant. 
With more tools available for oritatami programming including finite state control \cite{PchelinaSST2022} and the molecular implementation of the oritatami model within view, a class of tag systems or their functional enhancements with extra features that are not as computationally powerful as the Turing machine gets more significant. 

Including states enables even 1-tag systems to simulate Turing machines and they characterize the class of context-sensitive languages if all appended words are restricted to be of length at most 1, that is, the word never gets strictly longer than the initial one. This was shown a long time ago~\cite{zaiontz1976}, with such 1-tag systems with states being referred to as circular automata. 
A primary source of computability even under this restriction is that each ``cell'' of the input tape can be rewritten endlessly in an arbitrary manner. 
As it is experimentally not trivial to implement arbitrarily-rewritable media, it makes sense to suppose the \textit{freezing property}, under which a letter can be replaced with only a smaller letter according to some pre-determined order.

Many types of machines processing their input using some first-in-first-out storage have been investigated starting with queue automata and various restrictions thereof~\cite{Kutrib2018}. Such models use queues \emph{in addition to} their input tape and are generally quite powerful, whereas freezing 1-tag systems (Fr1TASS) are towards the lower end of computational complexity, as we will see. A model that is rather close in spirit to our subject is the iterated uniform finite transducer~\cite{KMMP22}. 
Such transducers can simulate freezing 1-tag systems in a straightforward manner, but are much more powerful in the general case due to the lack of the freezing property. 
However, limiting the number of so-called `sweeps' performed by these transducers by a function that is linear in the length of the input might produce systems that are similar in accepting power to our tag systems, although the bound on sweeps does not directly translate into the constant bound on the number of rewritings per position. Another recent computing device with a similarity to freezing tag systems is the one-way jumping automaton, which reads and erases letters on a circular tape~\cite{CFY16}, a behavior that can be simulated by freezing tag systems, but it is easy to see that those automata are strictly weaker than our current model.

In this paper, we explore basic properties of freezing 1-tag systems. In Section 2 we define the model and two different accepting modes borrowed from the theory of pushdown automata. We show that the accepting modes are equivalent in the general case. In Section 3 we start the study of the accepting power of the model. We can prove that it is between the regular and the context-sensitive languages and that it is not included in the context-free class, but at present we cannot show that the inclusion does not hold in the other direction either. Next, in Section 4 we show that the class of languages accepted is closed under Boolean operations and that with a simple idea one can trade off description size for time complexity when constructing systems for intersection or union. In Section 5 we study some fundamental decision problems such as emptiness, universality, equivalence, and show that they are not decidable by reduction to the Post Correspondence Problem. In Section 6 we discuss a restriction of Fr1TASS where the tape alphabet is the same as the input alphabet. In their restricted form these systems are much weaker and we can use a `computation flattening' argument to prove negative results about systems with accepting state mode. We conclude the paper with some remarks and suggestions for future research in Section 7.

\section{Preliminaries}
\label{sect:preliminaries}

Let $\Sigma$ be a finite alphabet and $\Sigma^*$ be a set of all words over $\Sigma$ including the empty word $\lambda$. 
The length of a word $w \in \Sigma^*$ is the number of letters that occur in $w$, and is denoted by $|w|$. 

A \textit{1-tag system with states} (1TASS) is a string rewriting system denoted by a tuple $(\Sigma, \Gamma, Q, q_0, F, \delta)$, where $\Sigma$ is a finite input alphabet, $\Gamma$ is a finite tape alphabet that includes $\Sigma$ as its subset, $Q$ is a finite set of states including the initial state $q_0$, $F \subseteq Q$ is a set of accepting states, and $\delta: Q \times \Gamma \to Q \times \Gamma^*$ is a transition function. 
The transition function is \textit{freezing} if, for all $q \in Q$, $a \in \Gamma$, and $(p, u) \in \delta(q, a)$, we have $|u| \le 1$ (length-non-increasing), and furthermore, either $u = \lambda$ or $u \le a$ according to some pre-defined order $\le$ over the elements of $\Gamma$. 
In this case, the 1TASS itself is also said to be freezing. We will refer to these systems as Fr1TASS.

A configuration of a freezing 1TASS $M = (\Sigma, \Gamma, Q, q_0, F, \delta, \tau)$ is a pair $(q, w)$ of the current state $q \in Q$ and a current word $w \in \Gamma^*$. 
Suppose this system is in the configuration $(p, a_1a_2 \cdots a_n)$ for some $p \in Q$, $n \ge 0$, and $a_1, a_2, \ldots, a_n \in \Gamma$. 
Then it can transfer to a configuration $(q, v)$ if $(q, b) \in \delta(p, a_1)$ and $v = a_2 \cdots a_n b$; in this case, we write $(p, a_1 \cdots a_n) \to_M (q, v)$. 
The reflexive and transitive closure of the relation $\to_M$ is denoted by $\to_M^*$. 
In the 1TASS, a word can be considered as being written on a cyclic tape along which a finite-state control moves in the clockwise direction while rewriting. 

\noindent{\bf Empty tape vs accepting state.} 
The model above can be introduced with or without deletion, which as we will see shortly, does not make a difference with respect to the accepting power. 
In the first case we allow transitions in which a letter is replaced by the empty word, effectively erasing the cell from the tape. 
In this scenario we can define acceptance conditions similar to pushdown automata: the machine accepts when the tape is empty (ET) or when the machine enters an accepting state (AS). 
Formally, starting from an initial configuration $(q_0, w)$, $M$ can accept an input word $w \in \Sigma^*$ in two different modes: \textit{by an accepting state} if $(q_0, w) \to_M^* (q_f, v)$ for some accepting state $q_f \in F$ and a word $v \in \Gamma^*$, while \textit{by the empty tape} if $(q_0, w) \to_M^* (q, \lambda)$ for some $q \in Q$. Where the distinction is necessary, the languages accepted by the system $M$ by accepting state and by empty tape will be denoted by $L(M)_{AS}$ and $L(M)_{ET}$, respectively. 
Note that the AS mode as defined here introduces a technical problem: a system in this mode either does not accept the empty word or it accepts every input. This is due to fact that accepting the empty word requires that the initial state is final, as well. However, then such a system accepts any input right away as it is already in a final state. Therefore, to simplify the presentation we allow AS mode systems to have a special transition from the initial state to a final state on reading $\lambda$ and require that such transition is only used when the input is empty. This allows stating our results in a more general form without emphasizing this caveat whenever talking about AS mode Fr1TASS. 

First we show that the two conditions lead to the same computational power, which simplifies our exposition further on as we will not have to specify the accepting mode. The following technical lemma states that AS mode machines do not actually need to erase any symbol from their tape.
\begin{lemma}\label{lem:non-erasing}
    For any Fr1TASS $A=(\Sigma, \Gamma_1, Q_1, q_1, F_1, \delta_1)$ there exists a Fr1TASS $B=(\Sigma, \Gamma_2, Q_1, q_1, F_1, \delta_2)$ such that $L(A)_{AS}=L(B)_{AS}$ and the transition function of $B$ does not erase symbols, that is, $\delta_2(q,a)=(q',\lambda)$ is not allowed for any $q,q'\in Q$ and $a\in \Gamma_2$ .
\end{lemma}
\begin{proof}
    The proof is straightforward. If $A$ does not erase symbols, then $B=A$ concludes the argument. If it does, then $B$ simulates all non-erasing transition of $A$ and for each erasing transition of $A$ of the form $\delta_1(q,a)=(q',\lambda)$ we set $\delta_2(q,a)=(q',\Box)$, where $\Box\in \Gamma_2\setminus\Gamma_1$ is a new symbol standing in for the positions erased by $A$. We set $\Box$ to be the smallest letter in $\Gamma_2$ and we add a loop $\delta_2(q,\Box)=(q,\Box)$ to each state $q$, ensuring that $B$ performs the same computations as $A$. 
\end{proof}

We now show that accepting modes ET and AS are equivalent.

\begin{lemma}[ET simulates AS]\label{lem:ETAS}
    For any Fr1TASS $A=(\Sigma, \Gamma_1, Q_1, q_1, F_1, \delta_1)$ there exists a Fr1TASS $B=(\Sigma, \Gamma_2, Q_1, q_1, F_1, \delta_2)$ such that $L(A)_{AS}=L(B)_{ET}$.
\end{lemma}
\begin{proof}
    In this case the ET machine will be almost identical to the AS one. To get the ET machine accepting the same language, we simply add erasing loops $\delta(f,a)=(f,\lambda)$ to all accepting states $f$ for all tape symbols $a$. By the definition of AS, when the machine reaches an accepting state, it halts, so we may assume w.l.o.g. that in $A$ there are no outgoing transitions from any accepting state, so adding the aforementioned transitions does not introduce non-determinism. If the AS machine reaches a final state by an input, the same input will take the ET machine into the same state, whereby it will erase all the remaining symbols from the tape. Conversely, by Lemma~\ref{lem:non-erasing} we may assume that $A$ never erases its tape, therefore the only words leading to an empty tape in $B$ will be the ones accepted by $A$.
\end{proof}

\begin{lemma}[AS simulates ET]
    For any Fr1TASS $A=(\Sigma, \Gamma_1, Q_1, q_1, F_1, \delta_1)$ there exists a Fr1TASS $B=(\Sigma, \Gamma_2, Q_1, q_1, F_1, \delta_2)$ such that $L(A)_{ET}=L(B)_{AS}$.
\end{lemma}
\begin{proof}
Similar to the argument in Lemma~\ref{lem:non-erasing}, we can replace erasing transitions $\delta_1(q,a)=(q',\lambda)$ with $\delta_2(q,a)=(q',\Box)$ for some newly introduced smallest letter $\Box$, and add $\Box$-labeled loops to all states. We create two copies of the tape alphabet of $A$, marked and unmarked. Because of this, when the computation begins, $B$ can mark the first letter to keep track of the start of the input. All operations on that symbol will be done with marked symbols. We duplicate the whole state diagram of $A$, such that the two copies of the states will `remember' whether some symbol other than $\Box$ was read since last passing the marked start. If the marked start is read twice with no non-$\Box$ symbol in between, it means that on the given input $A$ erased the tape, so $B$ will transition to an accepting state.
\end{proof}

\subsection{Examples}

\noindent{\bf Example 1.}
Even over a unary alphabet the ability to repeatedly read the tape allows freezing 1TASS to accept complex languages, such as the well-known non-context-free language $\{a^{2^n}\mid n\geq 0\}$. The tag system in Fig.~\ref{fig:a2n} is intuitive and is essentially the same as for iterated uniform finite state transducers (\cite{KMMP22}, Lemma 20). The system erases every second occurrence of $a$ which means that in each sweep it halves the length of the remaining tape content. Together with marking the first position with a special symbol at the start this allows Fr1TASS to process correct inputs in logarithmically many sweeps in the length of the input.
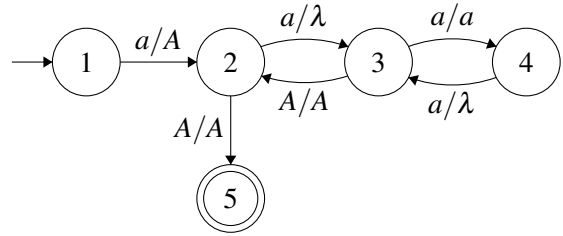
\begin{figure}
    \centering
    \begin{tikzpicture}[scale=0.15]
\tikzstyle{every node}+=[inner sep=0pt]
\draw [black] (8.4,-27.4) circle (3);
\draw (8.4,-27.4) node {$1$};
\draw [black] (21.2,-27.4) circle (3);
\draw (21.2,-27.4) node {$2$};
\draw [black] (34.3,-27.4) circle (3);
\draw (34.3,-27.4) node {$3$};
\draw [black] (21.2,-39.5) circle (3);
\draw (21.2,-39.5) node {$5$};
\draw [black] (21.2,-39.5) circle (2.4);
\draw [black] (47.4,-27.4) circle (3);
\draw (47.4,-27.4) node {$4$};
\draw [black] (1.8,-27.4) -- (5.4,-27.4);
\fill [black] (5.4,-27.4) -- (4.6,-26.9) -- (4.6,-27.9);
\draw [black] (11.4,-27.4) -- (18.2,-27.4);
\fill [black] (18.2,-27.4) -- (17.4,-26.9) -- (17.4,-27.9);
\draw (14.8,-26.9) node [above] {$a/A$};
\draw [black] (23.865,-26.042) arc (109.58784:70.41216:11.587);
\fill [black] (31.63,-26.04) -- (31.05,-25.3) -- (30.71,-26.24);
\draw (27.75,-24.87) node [above] {$a/\lambda$};
\draw [black] (21.2,-30.4) -- (21.2,-36.5);
\fill [black] (21.2,-36.5) -- (21.7,-35.7) -- (20.7,-35.7);
\draw (20.7,-33.45) node [left] {$A/A$};
\draw [black] (36.955,-26.021) arc (109.92269:70.07731:11.431);
\fill [black] (44.75,-26.02) -- (44.16,-25.28) -- (43.82,-26.22);
\draw (40.85,-24.84) node [above] {$a/a$};
\draw [black] (31.574,-28.635) arc (-72.4168:-107.5832:12.658);
\fill [black] (23.93,-28.64) -- (24.54,-29.35) -- (24.84,-28.4);
\draw (27.75,-29.73) node [below] {$A/A$};
\draw [black] (44.737,-28.762) arc (-70.34463:-109.65537:11.555);
\fill [black] (36.96,-28.76) -- (37.55,-29.5) -- (37.88,-28.56);
\draw (40.85,-29.94) node [below] {$a/\lambda$};
\end{tikzpicture}    
    \caption{Fr1TASS accepting the language $\{a^{2^n} \mid n\geq 0\}$.}
    \label{fig:a2n}
\end{figure}

\noindent{\bf Example 2.}
Our next example is the marked copy language $\{\#w\#w \mid w\in \{a,b\}^* \}$. It is well known that the language is not context-free by a simple application of the Bar-Hillel pumping lemma. A simple Fr1TASS as in Fig.~\ref{fig:ww} can accept this language by matching and erasing pairs of letters at the same distance from the two special markers $\#$ iteratively, accepting the language in linearly many sweeps.
\begin{figure}
    \centering
    \begin{tikzpicture}[scale=0.2]
\tikzstyle{every node}+=[inner sep=0pt]
\draw [black] (6.3,-28.2) circle (3);
\draw (6.3,-28.2) node {$I$};
\draw [black] (18.3,-28.2) circle (3);
\draw (18.3,-28.2) node {$S$};
\draw [black] (32.1,-17.1) circle (3);
\draw (32.1,-17.1) node {$A$};
\draw [black] (32.1,-38.9) circle (3);
\draw (32.1,-38.9) node {$B$};
\draw [black] (47.1,-17.1) circle (3);
\draw (47.1,-17.1) node {$1$};
\draw [black] (47.1,-38.9) circle (3);
\draw (47.1,-38.9) node {$2$};
\draw [black] (47.1,-28.2) circle (3);
\draw (47.1,-28.2) node {$M$};
\draw [black] (61.2,-28.2) circle (3);
\draw (61.2,-28.2) node {$3$};
\draw [black] (74.4,-28.2) circle (3);
\draw (74.4,-28.2) node {$4$};
\draw [black] (74.4,-38.9) circle (3);
\draw (74.4,-38.9) node {$Halt$};
\draw [black] (74.4,-38.9) circle (2.4);
\draw [black] (-0.3,-28.2) -- (3.3,-28.2);
\fill [black] (3.3,-28.2) -- (2.5,-27.7) -- (2.5,-28.7);
\draw [black] (9.3,-28.2) -- (15.3,-28.2);
\fill [black] (15.3,-28.2) -- (14.5,-27.7) -- (14.5,-28.7);
\draw (12.3,-28.7) node [below] {$\#/\$$};
\draw [black] (20.64,-26.32) -- (29.76,-18.98);
\fill [black] (29.76,-18.98) -- (28.83,-19.09) -- (29.45,-19.87);
\draw (23.49,-22.16) node [above] {$a/\lambda$};
\draw [black] (20.67,-30.04) -- (29.73,-37.06);
\fill [black] (29.73,-37.06) -- (29.4,-36.18) -- (28.79,-36.97);
\draw (23.43,-34.05) node [below] {$b/\lambda$};
\draw [black] (35.1,-17.1) -- (44.1,-17.1);
\fill [black] (44.1,-17.1) -- (43.3,-16.6) -- (43.3,-17.6);
\draw (39.6,-17.6) node [below] {$\#/\#$};
\draw [black] (30.777,-14.42) arc (234:-54:2.25);
\draw (32.1,-9.85) node [above] {$a/a,\mbox{ }b/b$};
\fill [black] (33.42,-14.42) -- (34.3,-14.07) -- (33.49,-13.48);
\draw [black] (35.1,-38.9) -- (44.1,-38.9);
\fill [black] (44.1,-38.9) -- (43.3,-38.4) -- (43.3,-39.4);
\draw (39.6,-39.4) node [below] {$\#/\#$};
\draw [black] (47.1,-20.1) -- (47.1,-25.2);
\fill [black] (47.1,-25.2) -- (47.6,-24.4) -- (46.6,-24.4);
\draw (47.6,-22.65) node [right] {$a/\lambda$};
\draw [black] (47.1,-35.9) -- (47.1,-31.2);
\fill [black] (47.1,-31.2) -- (46.6,-32) -- (47.6,-32);
\draw (47.6,-33.55) node [right] {$b/\lambda$};
\draw [black] (44.42,-29.523) arc (-36:-324:2.25);
\draw (39.85,-28.2) node [left] {$a/a,\mbox{ }b/b$};
\fill [black] (44.42,-26.88) -- (44.07,-26) -- (43.48,-26.81);
\draw [black] (50.1,-28.2) -- (58.2,-28.2);
\fill [black] (58.2,-28.2) -- (57.4,-27.7) -- (57.4,-28.7);
\draw (54.15,-28.7) node [below] {$\$/\$$};
\draw [black] (34.439,-15.227) arc (123.69978:14.54219:17.247);
\fill [black] (34.44,-15.23) -- (35.38,-15.2) -- (34.83,-14.37);
\draw (51.76,-12.92) node [above] {$a/\lambda$};
\draw [black] (60.646,-31.145) arc (-15.64871:-123.97459:17.227);
\fill [black] (34.43,-40.78) -- (34.81,-41.65) -- (35.37,-40.82);
\draw (51.64,-43.21) node [below] {$b/\lambda$};
\draw [black] (64.2,-28.2) -- (71.4,-28.2);
\fill [black] (71.4,-28.2) -- (70.6,-27.7) -- (70.6,-28.7);
\draw (67.8,-28.7) node [below] {$\#/\#$};
\draw [black] (74.4,-31.2) -- (74.4,-35.9);
\fill [black] (74.4,-35.9) -- (74.9,-35.1) -- (73.9,-35.1);
\draw (74.9,-33.55) node [right] {$\$/\$$};
\draw [black] (73.473,-31.052) arc (-20.97143:-159.02857:29.047);
\fill [black] (73.47,-31.05) -- (72.72,-31.62) -- (73.65,-31.98);
\draw (46.35,-50.2) node [below] {$\#/\#$};
\draw [black] (30.777,-36.22) arc (234:-54:2.25);
\draw (32.1,-31.65) node [above] {$a/a,\mbox{ }b/b$};
\fill [black] (33.42,-36.22) -- (34.3,-35.87) -- (33.49,-35.28);
\end{tikzpicture}
    \caption{Fr1TASS accepting the language $\{\#w\#w \mid w\in \{a,b\}^* \}$.}
    \label{fig:ww}
\end{figure}
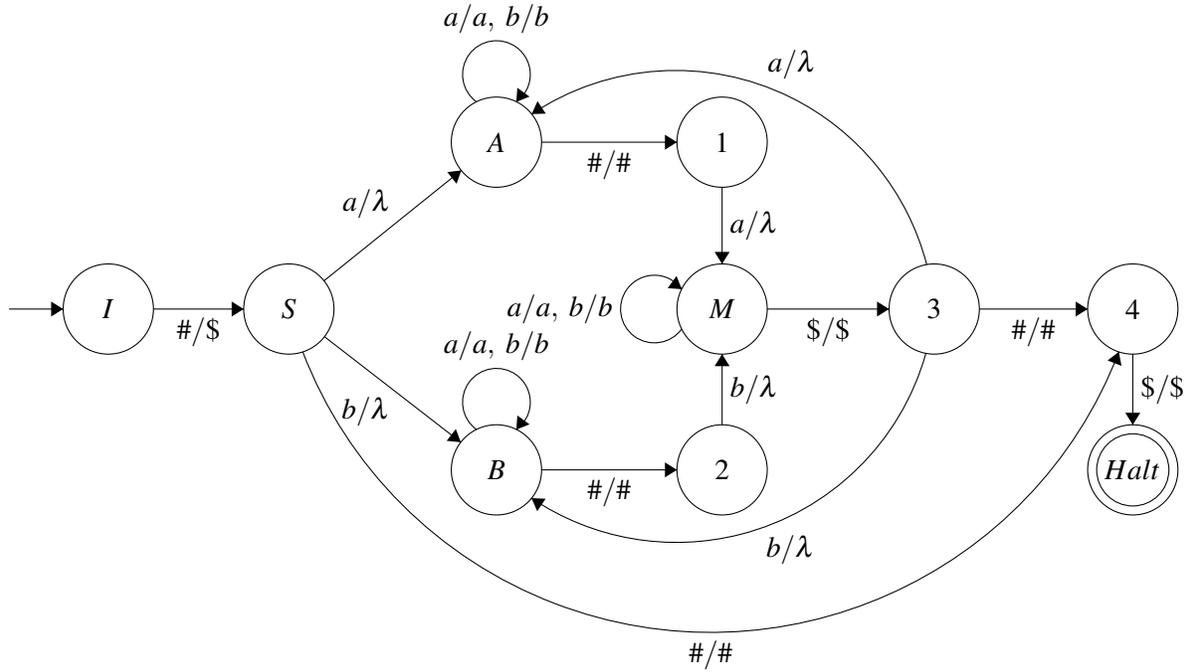

\noindent{\bf Example 3: accepting nondeterministic context-free languages.}
As will be detailed later, we were not able to prove that Fr1TASS cannot accept the language of palindromes, even though we conjecture that is the case. Our next idea was that perhaps such systems cannot even detect positions of the input at a certain ratio of the length from the beginning. Somewhat surprisingly, though, this proved to be false. Let us explain how a freezing 1TASS can detect the center of a given input. The idea can be adapted to detect positions at other linearly defined distances from the start.
The problem is formalized as follows: modify a given a freezing 1TASS so that, given an input $w = a_0 a_1 \cdots a_{n-1}$ of length $n \ge 0$, it can mark $a_{\lfloor n/2\rfloor}$ as a preprocess. 

Solving this problem is equivalent to marking the letters in the first half somehow. 
Let $k = \lfloor n/2 \rfloor$. 
The following algorithm first marks $k$ letters of $w$ (Step~\ref{item:marking}), and then move each mark ``rightward'' across the first letter $a_0$, which is distinguished from the other letters (Step~\ref{item:initial}). 

\begin{enumerate}
\item\label{item:initial} Mark the first letter as $s_0 a_1 \cdots a_{n-1}$. 
\item\label{item:marking} Mark every other letter as $s_0 \overline{a_1} a_2 \overline{a_3} \cdots$. 
	This results in $s_0 \overline{a_1} a_2 \overline{a_3} \cdots a_{n-2} \overline{a_{n-1}}$ if $n$ is even, or $s_0 \overline{a_1} a_2 \overline{a_3} \cdots \overline{a_{n-2}} a_{n-1}$ if $n$ is odd. 
\end{enumerate}
The following steps will be repeated until step 3 is unsuccessful, i.e., no overlined position is preceded by an unoverlined one. The read-write head starts at $s_0$, after finishing steps 1 and 2.

\begin{enumerate}
\setcounter{enumi}{2}
\item Find the first overlined letter, say $\overline{a_j}$, after at least one unoverlined one. 
\item Remove the overline as $\overline{a_j} \to a_j'$ (prime is must due to the freezing property). 
\item Scan the tape until $s_0$, which might be overlined. 
\item Overline the first unoverlined letter, which might have been primed. 
\end{enumerate}
One overline per iteration is shifted to the first half of the word. After repeating Steps 3-6 $k$ times, Step 3 fails to find an unoverlined letter prior to an overlined one; thus the system escapes from this loop. 
At this point, the $k$ overlines have been all moved to the left half of the word as $\overline{s_0} \overline{a_1'} \overline{a_2} \overline{a_3'} \cdots \overline{a_{k-1}'} a_k \cdots a_{n-1}$. 
The first unoverlined letter is $a_{\lfloor n/2 \rfloor}$.

\section{Power of Fr1TASS}

Let $\Sigma=\{1, \dots, k\}$ be the ordered alphabet of a 1TASS $A$ and consider an input $w=a_1\cdots a_n$, where $a_i\in \Sigma$. We can perceive the computation performed by $A$ as happening in `sweeps' on a circular tape. A sweep starts at the first position and ends when we reach that position again. Formally, the length of a \emph{sweep} of computation can be defined inductively. Let $r_1=n$ and then $r_i$ (for $i>1$) is defined as the length of the remaining input $w_i$ after $i-1$ sweeps, i.e., if $m=\sum_{\ell=1}^{i-1}r_\ell$, then for some $q\in Q$ we have $(q_0,w)\rightarrow^{m}(q,w_i)$ and we set $r_i=|w_i|$. Assuming that the computation stops, the freezing property of $A$ imposes that one of the following is true for each $i>1$:
\begin{enumerate}
    \item $r_i<r_{i-1}$.
    \item $w_{i-1}=b_1\cdots b_{r_{i}}$ and $w_i=c_1\cdots c_{r_i}$ with $c_j\leq b_j$ for each $j\in \{1,\dots, r_i\}$, and there exists some $j\in\{1,\dots, r_i\}$ such that $c_j<b_j$.
    \item $w_{i-1}=w_i$, but the sweeps $i-1$ and $i$ start in different states.
\end{enumerate}
From here, we can put an upper bound on the number of sweeps in a computation, and therefore an upper bound on the number of steps. 
Case 1 can happen at most $n$ times. 
From Case 3 we get that there can be at most $|Q|$ consecutive sweeps that do not change the tape. 
The number of times Case 2 occurs is bounded by the total number of possible rewritings. 
If the tape content is $w=b_1\cdots b_n$ with $b_i\in\{1,\dots, k\}$ for all $1\leq i\leq n$, then each letter can be rewritten to a smaller letter at most $k-1$ times, which means that Case 2 cannot occur more than $\sum_{\ell=1}^{n}b_\ell$ times, and $\sum_{\ell=1}^{n}(b_\ell-1)\leq n(k-1)$. 
This means that the number of sweeps performed on an input of length $n$ is $O(n)$. 
In each sweep we make at most $n$ steps to the right and a Turing machine simulating Fr1TASS would need to make at most $n$ steps to the left at the end of each sweep to return to the beginning. 
Thus we can conclude that 
the class of languages accepted by freezing 1TASS is included in DTIME($n^2$). 

The class of languages accepted by Fr1TASS strictly includes the class of regular languages. We can see that the inclusion is strict (even for unary languages) from the examples of the previous section. 
In order to show that the inclusion holds, we simulate a deterministic finite automaton by an AS mode Fr1TASS. 
The construction is simple but we need to bear in mind the fact that if the Fr1TASS reaches a final state reading a word $w$ consisting only of symbols of the input alphabet, that results in the system accepting any word from the language $w\Sigma^*$ according to the definition.

\begin{lemma}
For any regular language $R$ there exists a Fr1TASS $A$ such that $L(A)_{AS}=R$.    
\end{lemma}
\begin{proof}
    Let $M=(\Sigma, Q,q_0,F,\delta)$ be a deterministic finite automaton accepting $R$. We construct the Fr1TASS $A=(\Sigma, \Sigma\cup\{\Box\},Q\cup\{f_A\},q_0, \{f_A\}, \delta' )$, where $\Box\notin \Sigma$, $f_A\notin Q$, as follows. For each transition $\delta(q,a)=q'$ of the DFA, the Fr1TASS has a transition $\delta'(q,a)=(q',\Box)$. For each $f\in F$ we add the transitions $\delta'(f,\Box)=(f_a,\lambda)$. The system $A$ walks the same path in the transition diagram as $M$ does, but it replaces each letter by the $\Box$ symbol to mark it read. If reading the input takes the original system to a final state then the simulating Fr1TASS will have a $\Box$-labeled transition to a final state of its own. 
\end{proof}

Due to the AS and ET modes being equivalent, we can also construct an ET mode Fr1TASS for any regular language.
At the other extreme, it is easy to see that linear bounded automata can simulate Fr1TASS, which means that the class of languages accepted by these systems is included in the class of context-sensitive languages.
With respect to the other classes of the Chomsky hierarchy, we conjecture that the power of Fr1TASS is incomparable, but we lack the tools to show both sides of such statements. 
The examples from the previous section demonstrate that not every Fr1TASS language is context-free. We found the ability of Fr1TASS to mark the middle letter of a word counterintuitive, due to the fact that although (even one turn) pushdown automata can accept the related language of words with $a$ in the middle, it needs non-determinism to do so, as we will prove below. 
As Fr1TASS are deterministic, they cannot guess the middle and match the length of prefixes and suffixes. 
Nevertheless, nondeterminism is not required for Fr1TASS to mark the middle as shown in Example 3.

Consider $L_{a} = \{uav \mid |u|=|v| \}$. As we will show now, this language is not a deterministic context-free language. We will use the so-called DCFL pumping lemma below, due to Yu~\cite{YU89}.

\begin{lemma}\label{dcflpump}
Let $L$ be a deterministic context-free language. 
	Then there exists a constant $n$ for $L$ such that for any pair of words $w,w' \in L$ if
	\begin{enumerate}
		\item $w=xy$  and $w'=xz$, $|x|>n$, and
		\item first letter of $y$ = first letter of $z$,
	\end{enumerate}	
	then either 3. or 4. holds:
	\begin{enumerate}\setcounter{enumi}{2}
	\item there is a factorization $x=x_{1}x_{2}x_{3}x_{4}x_{5}$, with $|x_{2}x_{4}| \geq 1$ and $|x_{2}x_{3}x_{4}| \leq n$, 
	such that for all $i \geq 0$ we have that $x_{1}{x_{2}}^{i}x_{3}{x_{4}}^{i}x_{5}y$ and\  $x_{1}{x_{2}}^{i}x_{3}{x_{4}}^{i}x_{5}z$ are in L;
	\item there exist factorizations $x=x_{1}x_{2}x_{3}, y=y_{1}y_{2}y_{3}$ and $z=z_{1}z_{2}z_{3}$, with $|x_{2}| \geq 1$ and $|x_{2}x_{3}| \leq n, $
	such that for all $i \geq 0$ we have that $x_{1}{x_{2}}^{i}x_{3}y_{1}{y_{2}}^{i}y_{3}$ and $x_{1}{x_{2}}^{i}x_{3}z_{1}{z_{2}}^{i}z_{3}$ are in L.
	\end{enumerate}
\end{lemma}

\begin{theorem}
$L_{a}=\{uav \mid |u|=|v|\}$ is not a deterministic context-free language.
\end{theorem}
\begin{proof}
Suppose that $L_{a}$ were a DCFL and hence, that Lemma~\ref{dcflpump} applied. 
Let $n$ be the constant from the lemma and consider the words $x=b^{n+2}ab^{n+1}$, $y=a$ and $z=ab^{2n+4}$. 
It is easy to see that both $xy$ and $xz$ are in $L_{a}$ and long enough to meet the two conditions of the lemma, and the first letter of both $y$ and $z$ is $a$. Now we will show that assuming either conclusion of the lemma leads to a contradiction. First, suppose conclusion 3. is true. Depending on the factorization of $x$, we have the following cases:
\begin{enumerate}
    \item $x_1,x_2,x_4,x_5\in b^*$, $x_3\in b^*ab^*$: for $x_1x_2^0x_3x_4^0x_5y$ to be in $L_{a}$, we need $|x_2|=|x_4|$ and from the lemma we know they are not empty, so let $x_2=x_4=b^k$ for some positive $k\leq \frac{n}{2}$. However, then $x_1x_2^0x_3x_4^0x_5z=x_1x_2^0x_3x_4^0x_5ab^{2n+4}$, where the length of $x_1x_2^0x_3x_4^0x_5=2n+4-2k$, so the letter at the middle of the word is $b$, contradicting $x_1x_2^0x_3x_4^0x_5z\in L_{a}$.
    \item $x_1$ or $x_5$ contains $a$. In both cases $x_1x_2^0x_3x_4^0x_5y$ will result in a word with fewer $b$'s on one side of the first $a$ than the other, a contradiction.
    \item $x_2$ or $x_4$ has an $a$. In both cases $x_1x_2^0x_3x_4^0x_5y$ will have only one $a$, at the end of the word, a contradiction.
\end{enumerate}
Now suppose conclusion 4. is true. Since the factorization $x=x_1x_2x_3$ has the property $|x_2x_3|<n$ and $|x_2|\geq 1$, we know that $x_2=b^k$, for some positive $k\leq n$. However, this means that for any factorization $y=y_1y_2y_3$, the word $x_1x_2^0x_3y_1y_2^0y_3$ is of the form $b^{n+2}ab^{n+1-k}$ or $b^{n+2}ab^{n+1-k}a$. 
As neither of those is in $L_{a}$, because $k\geq 1$, we arrived at a contradiction again. 
Consequently, $L_a$ is not a deterministic context-free language.
\end{proof}

\section{Closure properties} 
We can show that the class of languages accepted by Fr1TASS forms a Boolean algebra, i.e., it is closed under union, intersection and complement. For the first two, we can adapt the classical construction used in the case of finite automata: the machine simulating the union/intersection of two others will have pairs of states representing the states of the starting machines and its alphabet will also consist of pairs of letters, to keep track of the tape of both simulated machines. 

\begin{theorem}
    The class of languages accepted by Fr1TASS is closed under union and intersection.
\end{theorem}
\begin{proof}
    Consider two languages, accepted by $A=(\Sigma, \Gamma_1, Q_1, q_1, F_1, \delta_1)$ and $B=(\Sigma, \Gamma_2, Q_2, q_2, F_2, \delta_2)$, respectively. We construct the system $$C=(\Sigma, \Sigma\cup(\Gamma_1\times\Gamma_2), Q_1\times Q_2, (q_1,q_2), F_1\times F_2, \delta)$$ accepting the language $L(A)\cap L(B)$ as follows. The computation of $C$ will simulate the computations of $A$ and $B$ in parallel, similar to the classical finite automaton construction. By Lemma~\ref{lem:non-erasing} we may assume that the systems $A$ and $B$ do not erase any symbols, so the length of the word on the tape is the same throughout the computation, making the parallel simulation possible. The difference is that here we have to observe the freezing property, so the ordering of the tape alphabet $\Gamma$ and the transition function $\delta$ need to be carefully defined. Let the total orderings of $\Gamma_1$ and $\Gamma_2$ be $\leq_1$ and $\leq_2$, respectively. Those two total orderings naturally define the partial order $\leq_{12}$ on $\Gamma_1\times\Gamma_2$ as $(a,b)\leq_{12} (c,d)$ if $a\leq_1 c$ and $b\leq_2 d$. By Szpilrajn's theorem~\cite{Szpilrajn}, every partial order has a linear extension, and we can efficiently construct such a linear order compatible with $\leq_{12}$ by any well-known topological sorting algorithm (e.g. Kahn's~\cite{Kahn}), since $\leq_1$ and $\leq_2$ are finite. We define the transition function of $C$ as $\delta((q_1,q_2),a)=((q_1',q_2'),(b_1,b_2))$ where
    $$((q_1',b_1),(q_2',b_2))=\begin{cases}
          (\delta_1(q_1,a),\delta_2(q_2,a)) & \textrm{ if } a\in \Sigma  \\
          (\delta_1(q_1,a_1),\delta_2(q_2,a_2)) & \textrm{ if } a=(a_1,a_2)\notin \Sigma
    \end{cases}.$$
    By the definition of $\delta$ we can be certain that the freezing property is preserved, that is, symbols of $\Gamma_2$ are rewritten respecting the linear extension of $\leq_{12}$. The proof for the union follows the same line with some small changes. Since the computations are done in parallel, it may happen that one of the machines gets stuck, i.e., it has no outgoing transition from its current state for the current input letter. However, if the other machine accepts, the input should be accepted. To handle such situations we introduce pairs of states $(q_1,\perp)$ and $(\perp,q_2)$ for all $q_1\in Q_1, q_2\in Q_2$, where having $\perp$ as one of the state components means the respective machine could not continue its computation. We can reach such states by transitions $\delta((q_1,q_2),a)=((q_1',\perp),b)$ when $\delta_1(q_1,a)=(q_1',b)$ and $\delta_2(q_2,a)$ is undefined, and then add transitions of the form $\delta((q_1',\perp),a)=((q_1'',\perp),(b,b))$ if $\delta_1(q_1',a)=(q_1'',(b,b))$ and their counterpart for the $(\perp,q_2)$ states. This way the state component tracking the stuck machine's state will be frozen while the other can continue the computation.
\end{proof}

\begin{theorem}
    The class of languages accepted by Fr1TASS is closed under complement.
\end{theorem}
\begin{proof}
For any Fr1TASS that halts on all inputs, it is enough to switch final and non-final states to accept the complement of its language. However, these systems may go into infinite loops, so a system accepting the complement needs to be able to detect that behavior. Each Fr1TASS can be completed with a `sink state', that is a state from which no other is reachable, and we can direct the transitions for any previously missing state-letter pair into that state. Additionally, we make $n+1$ copies of each state $p$, say, $p_1,\dots,p_{n+1}$, where $n$ is the number of states originally. For each $i$, the states indexed with $i$ are connected among them according to the original transition function, that is $\delta'(p_i,a)=(q_i,b)$ if $\delta(p,a)=(q,b)$, where $\delta$ and $\delta'$ are the transition functions of the original machine and of the machine for the complement, respectively. In the beginning, we mark the start of the input word with a special symbol to be able to keep track of it. Since the first symbol may need to be changed during the computation of the original machine, we add a marked copy of the original tape alphabet which will only be used to rewrite the first position. Whenever we read the start symbol in some state $p_i$, we continue on the $p_{i+1}$ states until one of two things happens:
\begin{itemize}
    \item We change one of the cells on the tape. In this case we continue the computation on the $p_1$ copies of the states until we reach the start mark.
    \item We reach the start mark from a $p_{n+1}$ state. This means that the machine made $n$ sweeps without changing any cell on the tape, so the original machine would go into an infinite loop. Instead, here we can simply transition to the sink state defined earlier.
\end{itemize}
The machine for the complement will have the newly introduced sink state as its only final state.
\end{proof}

Interestingly, the state complexity of intersection and union can be reduced at the expense of time complexity. 
This is because we can process the input twice instead of simulating both machines in parallel. 
First we process it according to the rules of the first machine, and then do so according to the rules of the second one. In order to do this, the number of states in the simulating system only needs to be the sum of the size of state sets of the two starting machines, instead of their product. Moreover, if we `recycle' the states, the size of the machine for the intersection/union need not increase beyond a constant plus the size of the larger machine participating in the intersection/union. When constructing the machine $C$ to accept $L(A)\cap L(B)$ (or $L(A)\cup L(B)$, respectively), we can reuse the states of the machine by having a tape alphabet with two tracks, say blue and red. Then, we can draw the red transitions completely independent of the blue transitions using the same states as vertices, therefore realizing $C$ on $\max\{|A|, |B|\}+k$ states. The additive constant term $k$ is needed, because after we finish simulating the first machine, we need to freeze the first track of the tape which requires some extra states to cycle through the input and mark each position frozen in the first track. 
We need to keep track of whether the first machine accepted or rejected the input. We can achieve this without extra states, though, by performing short-circuit evaluation: if the operation is intersection and the first machine rejects then we can reject right away; hence, if the simulation continues to the second machine, we know the first was accepted. 
The case for union can be treated analogously.

Regarding the regular operations concatenation and Kleene-star, the class of languages accepted by Fr1TASS is probably not closed, but we do not have the tools at present to prove that. In particular, we do not have any necessary conditions for a language to be accepted by Fr1TASS beyond the time complexity bound $O(n^2)$ mentioned earlier, and that bound is not enough to prove negative results regarding closure. The reason we think that the class is not closed under concatenation and Kleene-star is that in general such closure results require either non-deterministically guessing a decomposition of the input into factors of the constituent languages or the possibility of trying all possible decompositions. Neither option seems possible with Fr1TASS.

\section{Decision problems and minimal Fr1TASS}\label{sect:decision}
Using a construction similar to the freezing 1TASS accepting $\{\#w\#w \mid w\in\Sigma^*\}$, we will show how to reduce the Post Correspondence Problem (PCP) to the emptiness of freezing 1TASS languages. From that we can deduce that emptiness, universality ($=\Sigma^*$) and equivalence are undecidable for this model. We will argue that the undecidability of equivalence also strongly suggests that finding minimal freezing 1TASS for a given language cannot be algorithmically accomplished. 

An instance of PCP consists of two sets of words $U=\{u_1,\dots, u_n\}$ and $V=\{v_1,\dots, v_n \}$ and the instance is positive if there exists some finite sequence $k_1,\dots, k_\ell$ (a solution), with $k_i\in \{1,\dots, n\}$, such that $u_{k_1}\cdots u_{k_\ell}=v_{k_1}\cdots v_{k_\ell}$. It is a well-known fact that it is undecidable whether an instance of PCP is positive (\cite{pcp}). 

Let us fix the alphabet of the PCP instance as $\Gamma$, that is, $U,V\subseteq \Gamma^*$, and let $\Gamma' = \{1,\dots, n\}\cup \Gamma$. The alphabet of the machine will be $\Sigma=\{\#\} \cup \bigcup_{a\in \Gamma'}\{a,\overline{a}\}$. Choose any ordering of the alphabet such that $a<\overline{a}$ for each $a\in \Gamma'$. We construct a freezing 1TASS that accepts the language $$\{\# k_1\cdots k_\ell \# u_{k_1}\cdots u_{k_\ell} \# v_{k_1}\cdots v_{k_\ell} \mid u_{k_1}\cdots u_{k_\ell} = v_{k_1}\cdots v_{k_\ell}\},$$ where $k_i, u_i, v_i  \in (\bigcup_{a\in\Gamma'}\{\overline{a}\})^*$. The machine needs to check whether the input satisfies the following three conditions: (1) the middle part is indeed $u_{k_1}\cdots u_{k_\ell}$, (2) the last part is indeed $v_{k_1}\cdots v_{k_\ell}$ and (3) check whether $u_{k_1}\cdots u_{k_\ell} = v_{k_1}\cdots v_{k_\ell}$. As (2) can be done the same way as (1) and in parallel to it, and (3) has been illustrated before as the machine for $\{\#w\#w\}$, we will only detail (1). Figure~\ref{fig:pcp} illustrates the part of the system for checking (1). Since the factor between the second and third $\#$ and the one after the third $\#$ needs to be checked twice, first for (1) and (2), respectively, then for (3), all the input except the separators $\#$ needs to be marked by overlines at the beginning.

\begin{figure}
    \begin{center}
\begin{tikzpicture}[scale=0.15]
\tikzstyle{every node}+=[inner sep=0pt]
\draw [black] (6.9,-9) circle (3);
\draw (6.9,-9) node {$(1)$};
\draw [black] (20.2,-9) circle (3);
\draw (20.2,-9) node {$s_i$};
\draw [black] (36.3,-9) circle (3);
\draw (36.3,-9) node {$c_i$};
\draw [black] (58.3,-9) circle (3);
\draw (58.3,-9) node {$d_i$};
\draw [black] (46.2,-23.3) circle (3);
\draw (46.2,-23.3) node {$d_i$};
\draw [black] (30.8,-23.3) circle (3);
\draw (30.8,-23.3) node {$e_i$};
\draw [black] (75.1,-9) circle (3);
\draw (75.1,-9) node {$(2)$};
\draw [black] (20.2,-42.3) circle (3);
\draw (20.2,-42.3) node {$s_j$};
\draw [black] (36.3,-42.3) circle (3);
\draw (36.3,-42.3) node {$c_j$};
\draw [black] (58.3,-42.3) circle (3);
\draw (58.3,-42.3) node {$d_j$};
\draw [black] (46.2,-56.3) circle (3);
\draw (46.2,-56.3) node {$d_j$};
\draw [black] (30.8,-56.3) circle (3);
\draw (30.8,-56.3) node {$e_j$};
\draw [black] (9.9,-9) -- (17.2,-9);
\fill [black] (17.2,-9) -- (16.4,-8.5) -- (16.4,-9.5);
\draw (13.55,-8.5) node [above] {$\overline{i}/i$};
\draw [black] (18.877,-6.32) arc (234:-54:2.25);
\draw (20.2,-1.75) node [above] {$\overline{n}/\overline{n}$};
\fill [black] (21.52,-6.32) -- (22.4,-5.97) -- (21.59,-5.38);
\draw [black] (23.2,-9) -- (33.3,-9);
\fill [black] (33.3,-9) -- (32.5,-8.5) -- (32.5,-9.5);
\draw (28.25,-8.5) node [above] {$\#/\#$};
\draw [black][dashed] (39.3,-9) -- (55.3,-9);
\fill [black] (55.3,-9) -- (54.5,-8.5) -- (54.5,-9.5);
\draw (47.3,-9.5) node [below] {$\overline{u_i}\mbox{ }/\mbox{ }u_i$};
\draw [black] (34.977,-6.32) arc (234:-54:2.25);
\draw (36.3,-1.75) node [above] {$a/a$};
\fill [black] (37.62,-6.32) -- (38.5,-5.97) -- (37.69,-5.38);
\draw [black] (56.36,-11.29) -- (48.14,-21.01);
\fill [black] (48.14,-21.01) -- (49.04,-20.72) -- (48.27,-20.08);
\draw (52.8,-17.59) node [right] {$\#/\#$};
\draw [black] (43.2,-23.3) -- (33.8,-23.3);
\fill [black] (33.8,-23.3) -- (34.6,-23.8) -- (34.6,-22.8);
\draw (38.5,-22.8) node [above] {$\#/\#$};
\draw [black] (29.01,-20.89) -- (21.99,-11.41);
\fill [black] (21.99,-11.41) -- (22.06,-12.35) -- (22.86,-11.76);
\draw (26.08,-14.76) node [right] {$\overline{i}/i$};
\draw [black] (74.113,-11.832) arc (-22.26999:-121.94994:28.089);
\fill [black] (74.11,-11.83) -- (73.35,-12.38) -- (74.27,-12.76);
\draw (59.23,-28.53) node [below] {$\#/\#$};
\draw [black] (48.88,-21.977) arc (144:-144:2.25);
\draw (53.45,-23.3) node [right] {$\overline{a}/\overline{a}$};
\fill [black] (48.88,-24.62) -- (49.23,-25.5) -- (49.82,-24.69);
\draw [black] (56.977,-6.32) arc (234:-54:2.25);
\draw (58.3,-1.75) node [above] {$\overline{a}/\overline{a}$};
\fill [black] (59.62,-6.32) -- (60.5,-5.97) -- (59.69,-5.38);
\draw [black] (28.12,-24.623) arc (-36:-324:2.25);
\draw (23.55,-23.3) node [left] {$n/n$};
\fill [black] (28.12,-21.98) -- (27.77,-21.1) -- (27.18,-21.91);
\draw [black] (1.2,-9) -- (3.9,-9);
\fill [black] (3.9,-9) -- (3.1,-8.5) -- (3.1,-9.5);
\draw [black] (23.2,-42.3) -- (33.3,-42.3);
\fill [black] (33.3,-42.3) -- (32.5,-41.8) -- (32.5,-42.8);
\draw (28.25,-42.8) node [below] {$\#/\#$};
\draw [black][dashed] (39.3,-42.3) -- (55.3,-42.3);
\fill [black] (55.3,-42.3) -- (54.5,-41.8) -- (54.5,-42.8);
\draw (47.3,-42.8) node [below] {$\overline{u_j}/u_j$};
\draw [black] (56.977,-39.62) arc (234:-54:2.25);
\draw (58.3,-35.05) node [above] {$\overline{a}/\overline{a}$};
\fill [black] (59.62,-39.62) -- (60.5,-39.27) -- (59.69,-38.68);
\draw [black] (34.977,-39.62) arc (234:-54:2.25);
\draw (36.3,-35.05) node [above] {$a/a$};

\fill [black] (17.2,-42.42) -- (16.6,-41.6) -- (16.3,-43);
\draw [black] (7,-12) arc (150:247:21.5);
\draw (0,-22.5) node [above] {$\overline{j}/j$};

\fill [black] (37.62,-39.62) -- (38.5,-39.27) -- (37.69,-38.68);
\draw [black] (17.449,-41.134) arc (274.76512:-13.23488:2.25);
\draw (14.03,-36.7) node [above] {$\overline{n}/\overline{n}$};
\fill [black] (19.45,-39.41) -- (19.88,-38.57) -- (18.89,-38.65);
\draw [black] (29.34,-25.92) -- (21.66,-39.68);
\fill [black] (21.66,-39.68) -- (22.49,-39.23) -- (21.61,-38.74);
\draw (24.84,-31.6) node [left] {$\overline{j}/j$};
\draw [black] (43.2,-56.3) -- (33.8,-56.3);
\fill [black] (33.8,-56.3) -- (34.6,-56.8) -- (34.6,-55.8);
\draw (38.5,-55.8) node [above] {$\#/\#$};
\draw [black] (48.88,-54.977) arc (144:-144:2.25);
\draw (53.45,-56.3) node [right] {$\overline{a}/\overline{a}$};
\fill [black] (48.88,-57.62) -- (49.23,-58.5) -- (49.82,-57.69);
\draw [black] (30.144,-59.215) arc (15.0577:-272.9423:2.25);
\draw (23.97,-62.04) node [below] {$n/n$};
\fill [black] (28.09,-57.55) -- (27.18,-57.28) -- (27.44,-58.24);
\draw [black] (28.99,-53.91) -- (22.01,-44.69);
\fill [black] (22.01,-44.69) -- (22.1,-45.63) -- (22.89,-45.03);
\draw (26.07,-47.9) node [right] {$\overline{j}/j$};
\draw [black] (56.34,-44.57) -- (48.16,-54.03);
\fill [black] (48.16,-54.03) -- (49.06,-53.75) -- (48.31,-53.1);
\draw (52.8,-50.75) node [right] {$\#/\#$};
\draw [black] (76.654,-11.565) arc (28.61712:-114.86546:33.268);
\fill [black] (76.65,-11.57) -- (76.6,-12.51) -- (77.48,-12.03);
\draw (72.26,-51.71) node [right] {$\#/\#$};
\draw [black] (27.81,-56.076) arc (-97.73263:-237.00462:24.922);
\fill [black] (17.59,-10.48) -- (16.65,-10.49) -- (17.19,-11.33);
\draw (6.09,-37.21) node [left] {$\overline{i}/i$};
\end{tikzpicture}
\end{center}
    \caption{The parts of the 1TASS for matching the first portion of the input containing the numbers, $k_1\cdots k_\ell$, to the second portion, $u_{k_1}\cdots u_{k_\ell}$. If the number read is $i$, that is, the symbol $\overline{i}$, then we continue from $s_i$, if it is $\overline{j}$ then continue from $s_j$, and so on. Then, the machine looks for the $\#$ symbol after which it ignores the already matched parts of $u_{k_1}\cdots u_{k_\ell}$. Finding the first unmatched symbols, it matches them against $u_i$, after which it returns to the beginning and reads the next number. }
    \label{fig:pcp}
\end{figure}
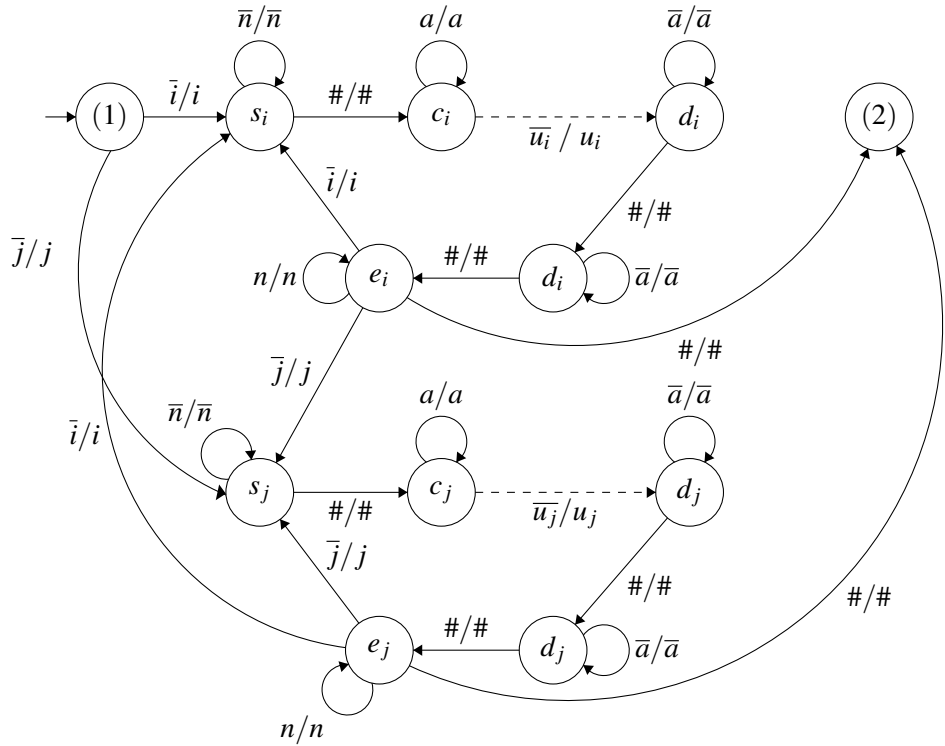

\begin{enumerate}
    \item Checking whether a word equals $u_i=x_1\cdots x_m$ is easy: we set up $m+1$ states $q_0,\dots,q_m$ such that $\delta(q_i,\overline{x_{i}})=(q_{i+1},x_i)$. For all $a\neq x_i$, the state $q_i$ has no outgoing transitions, therefore immediately rejecting the input on reading those letters.
    \item We read the first unmatched number after the first $\#$, say $i$. We move without changing over all following symbols until we reach the next $\#$. Then move over all matched symbols, i.e., symbols without overline. From the first symbol with overline, we match $u_i$ to the input, as above. If successful, move over all following symbols until we meet the second $\#$. Move over all symbols without overline and start the process again.
\end{enumerate}

This 1TASS will accept the solutions to the PCP instance, if any. Since PCP is undecidable, deciding whether the language accepted by a freezing 1TASS is empty, is also undecidable. This means that language equivalence is undecidable: if we let freezing 1TASS $A$ and $B$ be such that $A$ does not accept any input, while $B$ accepts the solutions of a PCP instance, then deciding equivalence amounts to deciding whether the PCP instance is positive. Similarly, if we let $L(C)=\Sigma^*\setminus L(B)$, where $B$ accepts the solutions to a PCP instance, then a decision algorithm that could tell whether $L(C)=\Sigma^*$, would decide whether the PCP instance has solutions, so universality is also undecidable.\\

\noindent{\bf Minimization of 1TASS. } From the undecidability of the language equivalence, we can draw certain conclusions regarding the minimization of such systems. Say we define minimal 1TASS as ones having the fewest number of states and/or transitions. We instantly get that the following statements cannot both be true, otherwise equivalence would be decidable by the same isomorphism checking method as for DFA:
\begin{enumerate}
    \item For each freezing 1TASS $A$ there is a unique (up to renaming the states) minimal 1TASS $B$ with $L(A)=L(B)$.
    \item There is an algorithm to find for each freezing 1TASS $A$ a minimal freezing 1TASS $B$ with $L(A)=L(B)$.
\end{enumerate}

If 1. holds then we cannot find the unique minimal system. Therefore we could assume that 1. does not hold and try to devise and algorithm for finding \emph{a} minimal system.

Another possibility is to define minimal systems more tightly, in which case minimization algorithms might exist. We suggest the following possible alternative definitions for a freezing 1TASS $A$ to be minimal:\\
1. $A$ does not contain strongly equivalent states, i.e., states $p,q$ such that $\delta(p,a)=\delta(q,a)$ for all $a\in\Sigma$. This case is straightforward to deal with along with any unreachable states, but yields little information about the similarity of Fr1TASS.\\
2. No proper subset of the system (removed transitions or states) accepts $L(A)$. Even the question whether minimality is decidable under this definition is nontrivial, let alone finding such a minimal system for a given Fr1TASS.\\

\section{Fr1TASS with no auxiliary symbols}\label{sect:auxiliary}
In this section we look at Fr1TASS that cannot have `auxiliary' symbols (which cannot appear in the input, but can occur on the tape during the computation), that is, $\Sigma=\Gamma$. This type of restriction leads to dramatic changes in computing power even in the case of machines that can rewrite cells arbitrarily many times~\cite{HEMASPAANDRA2005163}.
For these systems we can show that AS mode is incomparable to ET mode. Simulating AS with ET mode as done in the general case in Lemma~\ref{lem:ETAS} does not work. This is because we can no longer assume that the AS mode machines do not erase their tapes, as the technique used in Lemma~\ref{lem:non-erasing} is not applicable anymore due to the lack of symbols that can stand in for erased ones. In fact, unary languages provide the proof that under the no-auxiliary-symbols restriction, AS and ET are incomparable.

\begin{lemma}
    For each Fr1TASS $A=(\{a\},\{a\},Q,q_0, F, \delta)$, the language $L(A)_{AS}$ accepted with accepting state is of the form $\{a^n \mid n\geq k\}$ for some fixed $k$, and the language $L(A)_{ET}$ accepted with empty tape is either finite or equal to $a^*$.
\end{lemma}
\begin{proof}
    Just like in the case of deterministic finite automata, such systems $A$ have a transition diagram of a loop with a `handle', due to determinism. There are two types of transitions possible: erasing, that is, $\delta(q,a)=(q',\lambda)$ and non-erasing, that is, $\delta(q,a)=(q',a)$. In AS mode the system accepts and halts as soon as it reaches a final state which means that any input longer than the distance from the initial state to the first final state will be accepted. In fact, any input with more letters than the number of erasing transitions on the path from initial to final state will also be accepted. In ET mode for the system to accept anything other than the empty word, it needs to have at least one erasing transition. If the only such transitions are on the `handle', then the accepted language is finite, since the system can only erase finitely many symbols from the tape. If there is an erasing transition in the loop, then all inputs on which the machine reaches the loop will be accepted, since the machine will keep looping until all letters are erased.
\end{proof}

We can also show that AS mode cannot be strictly stronger than ET mode when the tape alphabet is at least binary. Consider the language $L_{ab}=\{w \mid |w|_b\leq |w|_a \leq |w|_b+1\}$. A machine in ET mode can easily accept this language by reading an $a$, erasing it, moving to the right until it finds a corresponding $b$, erasing it and iterating this process (Fig.~\ref{fig:abnoaux}). However, using a `computation flattening' argument we can prove that a machine in AS mode cannot accept this language. 

\begin{figure}
    \centering
    \begin{tikzpicture}[scale=0.15]
\tikzstyle{every node}+=[inner sep=0pt]
\draw [black] (12.3,-27.4) circle (3);
\draw (12.3,-27.4) node {$1$};
\draw [black] (27.7,-27.4) circle (3);
\draw (27.7,-27.4) node {$2$};
\draw [black] (5.7,-27.4) -- (9.3,-27.4);
\fill [black] (9.3,-27.4) -- (8.5,-26.9) -- (8.5,-27.9);
\draw [black] (14.798,-25.753) arc (116.12306:63.87694:11.815);
\fill [black] (25.2,-25.75) -- (24.7,-24.95) -- (24.26,-25.85);
\draw (20,-24.05) node [above] {$a/\lambda$};
\draw [black] (25.151,-28.967) arc (-65.37436:-114.62564:12.361);
\fill [black] (14.85,-28.97) -- (15.37,-29.76) -- (15.79,-28.85);
\draw (20,-30.59) node [below] {$b/\lambda$};
\draw [black] (26.377,-24.72) arc (234:-54:2.25);
\draw (27.7,-20.15) node [above] {$a/a$};
\fill [black] (29.02,-24.72) -- (29.9,-24.37) -- (29.09,-23.78);
\draw [black] (10.977,-24.72) arc (234:-54:2.25);
\draw (12.3,-20.15) node [above] {$b/b$};
\fill [black] (13.62,-24.72) -- (14.5,-24.37) -- (13.69,-23.78);
\end{tikzpicture}
    \caption{Fr1TASS accepting $\{w \mid |w|_b\leq |w|_a \leq |w|_b+1\}$ without auxiliary symbols in ET mode.}
    \label{fig:abnoaux}
\end{figure}
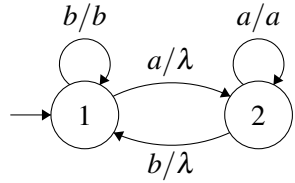

\begin{figure}
    \centering
    \includegraphics[scale=.5]{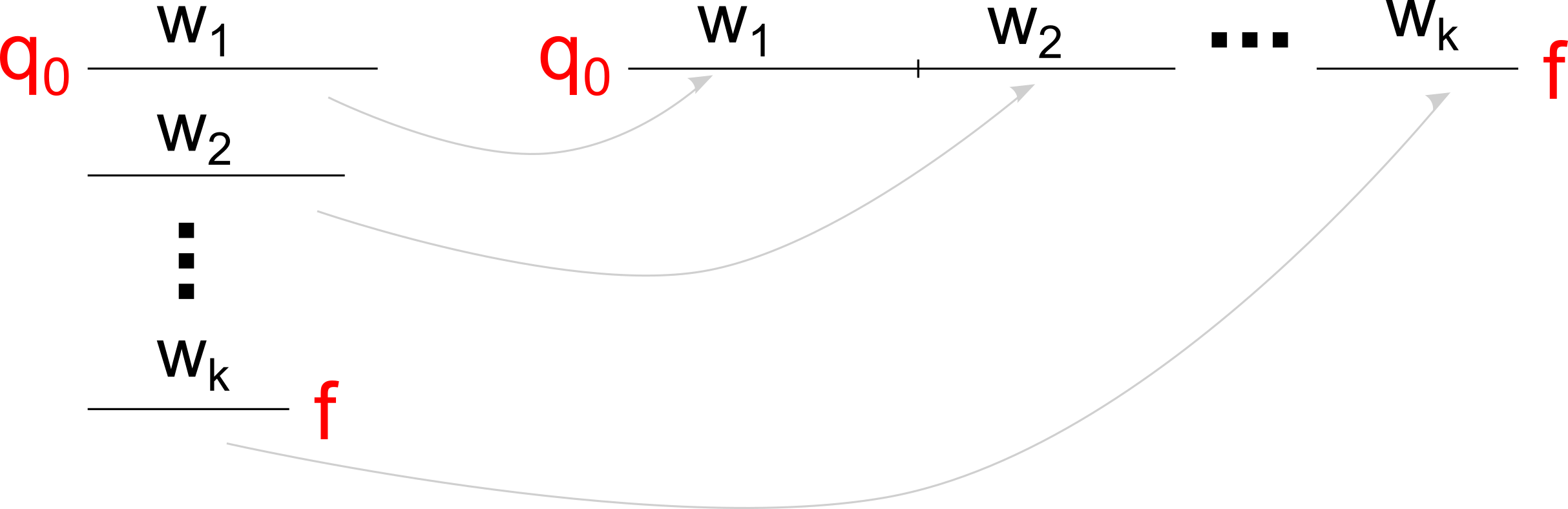}
    \caption{Flattening the computation of an AS mode system with no auxiliary symbols.}
    \label{fig:flattening}
\end{figure}

\begin{lemma}
    There is no Fr1TASS $A=(\Sigma, \Sigma, Q, q_0,F,\delta)$ such that $L(A)_{AS}=L_{ab}$.
\end{lemma}
\begin{proof}
    Assume there is a Fr1TASS $A$ as above that accepts $L_{ab}$. Take any $w\in L_{ab}$ and let the word on the tape in sweep $i$ of the accepting computation on $w$ be $w_i$, as defined in the preliminaries. 
    As the word is accepted, there are finitely many, say $k$, sweeps. 
    If we concatenate the words in the sweeps, we get $w'=w_1\cdots w_k$. The obtained word $w'$ is a valid input word, because the system can only use the symbols of the input alphabet. 
    On the input $w'$, the system $A$ reaches the same accepting state as on the input $w$; thus $w'$ is accepted, and the computation requires a single sweep. Moreover, the same final state is reached on input $w'aa$, too. However, this is a contradiction, since $w'aa$ cannot have the required numbers of letter occurrences if $w'$ did.
\end{proof}

\noindent{\bf Language of palindromes.}
A very challenging problem is whether the Fr1TASS model can accept the language of palindromes over a non-unary alphabet. Intuitively the model should not be able to accept such a language for the reason described below, but we do not have a proof for this due to the lack of applicable necessary conditions. To verify whether the input is a palindrome, a machine would need to match pairs of letters at the same distance from the middle or from the start and end, respectively. Moreover, the matched pairs would need to be marked to keep track of which parts still need matching. However, in this model, we can only mark symbols \emph{after} a pattern has been identified. This means that if we start matching letters at the same distance from the middle, then the machine could not guess which is the next unmatched letter in the left half. Conversely, if the machine matches pairs based on their distance from the left and right end, respectively, then it could not guess the next unmatched letter in the right half of the input. 

For the restricted model with no auxiliary symbols in AS mode, we can prove that the language of palindromes cannot be accepted, by using the computation flattening argument seen earlier. Let $L_{pal}$ denote the language of palindromes over the binary alphabet $\{a,b\}$, that is, $L_{pal}=\{w\in \{a,b\}^* \mid w=w^R\}$ where $w^R$ is the reverse of $w$, that is, if $w=a_1\cdots a_n$ then $w^R=a_n\cdots a_1$.

\begin{theorem}
    For any Fr1TASS $A=(\Sigma, \Gamma, Q, q, F, \delta)$ we have $L_{pal} \neq L(A)_{AS}$.
\end{theorem}
\begin{proof}
    Suppose that there were a Fr1TASS $A=(\Sigma, \Gamma, Q, q, F, \delta)$ that accepts $L_{pal}$ in accepting state mode. Consider a long palindrome of the form $a^n bw ba^n$, where $n>|Q|$, which is accepted by the system in $k$ sweeps. Again, let $w_i$ denote the word on the tape at the beginning of sweep $i$. Concatenating those words yields the valid input $w'=w_1\cdots w_k$, which will be accepted with the same transitions as $w$, but all in one sweep. Since $w'$ is accepted, it must be a palindrome by our assumption, which means that its suffix must be $ba^n$. Due to the fact that $n$ is larger than the number of states in $A$, while reading the suffix $a^n$, the system must enter some state $p$ more than once, reading $a^\ell$ for some $\ell\geq 1$, between the first two traversals of $p$. However, this means that the system accepts also words of the form $w'a^{i\ell}$, for all $i\geq 0$. This results in a contradiction for $i=1$, because the suffix of $w'a^\ell$ is $ba^{n+\ell}$ while its prefix is $a^nb$. 
    Thus, non-palindromes would be also accepted by $A$.
\end{proof}

\section{Concluding remarks}
Apart from the decision problems in Section~\ref{sect:decision}, our results have mostly been positive. To establish the limits of the accepting power of Fr1TASS we need negative results separating Fr1TASS languages from other language classes. Although these systems can process symbols in the same position repeatedly, we think that the freezing property allows some form of a pumping lemma, perhaps in combination with a computation flattening argument seen in Section~\ref{sect:auxiliary}. Obtaining such a tool seems quite challenging and will be our main focus in future studies on the topic.

Perhaps with a tool as described above or adapting the Kolmogorov complexity argument of Li et al.~\cite{LLV92}, one could prove that the language of palindromes is not a Fr1TASS language. This is intuitively a fundamental limitation of such systems with first-in-first-out (FIFO) nature of processing and such questions have proved interesting in their own right with respect to other FIFO style models~\cite{M15}. If indeed palindromes cannot be accepted with Fr1TASS, then the language class is in some sense a natural counterpart of the class of context-free languages: membership is decidable efficiently and it contains the FIFO-like copy language instead of the LIFO-like palindromes. Interestingly, if one allows Fr1TASS to have non-determinism, then this FIFO limitation seems to vanish: such Fr1TASS could now guess which letters form pairs at equal distance from a reference point (middle or the ends) and verify the guess by marking symbols. The power of nondeterministic Fr1TASS is another topic worth further exploration in our opinion.

Finally, we would like to mention a computational complexity aspect that could be investigated with respect to Fr1TASS. The computation happens in sweeps and those sweeps are a natural resource to measure as the complexity of a given computation. Based on the amount of this resource used, one can define and study asymptotic complexity classes similarly to the case of iterated uniform finite transducers~\cite{KMMP22} and one-way jumping finite automata~\cite{FMW22}.

\bibliographystyle{eptcs}
\bibliography{fr1tass}

\end{document}